\documentclass[11pt,a4paper]{llncs}
\usepackage[margin=1in]{geometry}

\usepackage{graphicx}
\usepackage{amsmath}
\usepackage{amssymb}
\usepackage{amsfonts}

\usepackage{url}

\newcommand{\SA}{\ensuremath{\mathsf{SA}}}
\newcommand{\BWT}{\ensuremath{\mathsf{BWT}}}
\newcommand{\from}{\ensuremath{\mathit{from}}}

\newcommand{\set}[1]{\ensuremath{\{ #1 \}}}
\newcommand{\Exp}[1]{\ensuremath{\mathrm{E}\left[ #1 \right]}}
\newcommand{\abs}[1]{\ensuremath{\lvert #1 \rvert}}

\newcommand{\onebit}{$1$\nobreakdash-bit}
\newcommand{\zerobit}{$0$\nobreakdash-bit}

\title{Indexing Finite Language Representation of Population Genotypes}

\author{Jouni Sir{\'e}n\thanks{Funded by the Finnish Doctoral Programme in Computational Sciences, 
Academy of Finland project 1140727 
and the Nokia Foundation.
}
\and
Niko V{\"a}lim{\"a}ki
\thanks{Funded by Helsinki Graduate School in Computer Science and Engineering 
and Finnish Centre of Excellence for Algorithmic Data Analysis Research 
}
\and
Veli M{\"a}kinen
}

\institute{Helsinki Institute for Information Technology (HIIT) \& \\
Department of Computer Science, University of Helsinki, Finland\\
\email{\{jltsiren,nvalimak,vmakinen\}@cs.helsinki.fi}}

\begin{document}

\maketitle

\pagestyle{plain} 

\begin{abstract}
With the recent advances in DNA sequencing, it is now possible to have complete genomes of individuals sequenced and assembled. This rich and focused genotype information can be used to do different population-wide studies, now first time directly on whole genome level.
We propose a way to index population genotype information together with the 
complete genome sequence, so that one can use the index to efficiently align a given sequence to the genome with all plausible genotype recombinations taken into account. This is achieved through converting a multiple alignment of individual genomes into a finite automaton recognizing all strings that can be read from the alignment by switching the sequence at any time. The finite automaton is indexed with an extension of Burrows-Wheeler transform to allow pattern search inside the plausible recombinant sequences. The size of the index stays limited, because of the high similarity of individual genomes. 
The index finds applications in variation calling and in primer design.
On a variation calling experiment, we found about 1.0\% of matches to novel recombinants just with exact matching, and up to 2.4\% with approximate matching.
\end{abstract}

\section{Introduction}\label{sect:introduction}

Due to the advances in DNA sequencing \cite{Met10}, it is now possible to have complete genomes of individuals sequenced and assembled. Already several human genomes have been sequenced \cite{Venetal01,HGC01,LSNetal07,Wheetal08,Lietal10} and it is almost a routine task to resequence individuals by aligning the high-throughput short DNA reads to the reference \cite{FB09}. This rich and focused genotype information, together with the more global genotype information (common single nucleotide polymorphisms (SNPs) and other variations) created using earlier techniques, can be combined to do different population-wide studies, now first time directly on whole genome level. 

We propose a novel index structure that simultaneously represents and extrapolates the genotype information present in the population samples. The index structure is built on a given multiple alignment of individual genomes, or alternatively for a single reference sequence and set of SNPs of interest. The index structure is capable of aligning a given pattern to any path taken along the multiple alignment, as illustrated in Figure~\ref{fig:multiplealignment}.\footnote{It is also possible to limit to only those paths that contain recombination hotspots \cite{Mye05}, to avoid matching to too rare recombinants. For brevity, we cover here only the unrestricted case.}

\begin{figure}[t]
\centerline{\includegraphics{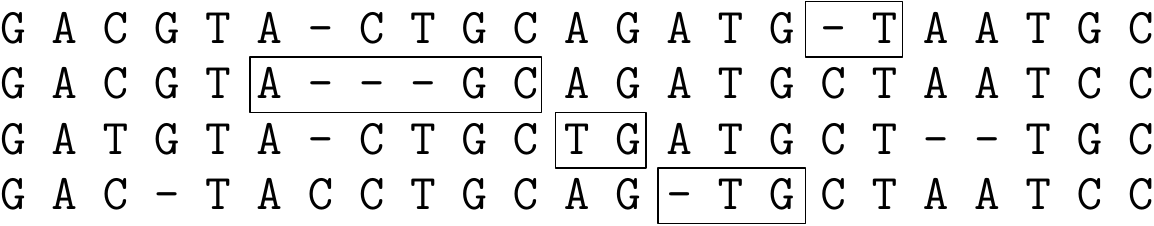}}
\caption{Pattern \texttt{AGCTGTGT} matching the multiple alignment when allowing it to change row when necessary.} 
\label{fig:multiplealignment}
\end{figure}

To build the index, we first create a finite automaton recognizing all paths 
through the multiple alignment, and then generalize Burrows-Wheeler transform 
(BWT) \cite{Burrows1994} \nobreakdash-based self-index structures 
\cite{Navarro2007} to index paths in labeled graphs. The backward search routine 
of BWT\nobreakdash-based indexes generalizes to support exact pattern search over 
the labeled graph in $O(m)$ time, for pattern of length $m$. On general labeled 
graphs, such index can take exponential space, but on graphs resulting from finite 
automaton representation of multiple alignment of individual genomes, 
the space is expected to stay limited.

Applications for our index include the following:
\begin{itemize}

\item \emph{SNP calling}. We can take the known SNPs into account already 
in the short read alignment phase, instead of the common pipeline of alignment, variation 
calling, and filtering of known SNPs. This allows more accurate alignment, as the 
known SNPs are no longer counted as errors, and the matches can represent novel 
recombinants not yet represented in the database. 
Hence, we can expect to output less 
false positive candidates for the postprocessing steps such as for the common realignment 
step that creates a consensus for indels by multiple alignment of nearby mapped reads.
\item \emph{Probe/primer design.} When designing probes for microarrays or primers 
for PCR, it is important that the designed sequence does not occur even 
approximately elsewhere than in the target. Our index can provide 
approximate search not only against all 
substrings, but also against plausible recombinants, and hence the design can be 
made more selective.
\item \emph{Large indel calling}. After short read alignment, the common approach 
for detecting larger indels is to study the uncovered regions in the reference 
genome. 
If there is a deletion in the donor, there should be an almost equal size 
gap in the mapping result. If there is an insertion in the donor, there should be 
a pair of positions in the reference covered by no single read (in the perfect 
world).  
We can model deletions with our index by adding an edge to the automaton 
skipping each plausible deleted area. For insertions, we can apply \emph{de novo} 
sequence assembly with the unmapped reads 
(and pair of fixed start and end 
sequences representing the prefix and suffix of the predicted insert position) 
to generate plausible insertions, and add these as paths to the automaton. 
(or even 
multiple paths read from an overlap graph, without fixing the assembly). 
Realigning all the reads with the index built after adding the plausible deletions 
and insertions to the automaton then gives a voting result to call for the large 
variants.
\item \emph{Reassembly of donor genome}. Continuing from variation calling, 
one can use the realignment of the reads to the refined 
automaton to give a probability for each edge. It is then easy to extract, for 
example, the most probable path through the automaton as a consensus for the 
donor.
\end{itemize}

We made some feasibility experiments on 
SNP calling problem. We created our index 
on a multiple alignment of four instances of the 76 Mbp human chromosome 18. The 
total size of the index was about 67 MB. Aligning a set of 10 million Illumina 
Solexa reads of length 56 took 18 minutes, and about 1.0\% of exact matches were 
novel recombinants not found when indexing each chromosome instance separately 
(see Sect.~\ref{sect:experiments}). Leaving these exact matches out from  
variation calling reduces the number of novel SNPs from 4203 to 1074.
With approximate matching, the proportion of 
novel recombinants increased to 2.4\%.
 
\paragraph{Related work and extensions.}

\emph{Jumping alignments} of protein sequences were studied in \cite{SRS02} as
an alternative to profile Hidden Markov Models.
They showed how to do local alignment across a multiple alignment of protein family so that 
jumping from one sequence to another is associated with a penalty. 
This gives an alternative to 
the Markov model approaches to decide whether a given protein is part of the protein family. 
The computation requires dynamic programming through the multiple alignment.
We study the same problem but from a different angle; we 
provide a compressed representation of the multiple alignment with 
an efficient way to support pattern matching.

Calling of large indels similar to our approach was studied in \cite{Alb10}.
One difference is that they manage to associate probabilities to the putative genotypes, resulting into a more reliable calling of likely variants. 
However, their indexing part is tailored for this specific problem, whereas we develop a more systematic approach that can be generalized to many directions. 
For example, we can take the probabilities into account and index only paths with high enough probabilities, 
to closely simulate their approach (see Sect.~\ref{sect:discussion}).

Our work builds on the self-indexing scenario \cite{Navarro2007}, and more specifically is an extension of 
the \emph{XBW transform} \cite{Ferragina2009b} that is an index structure for labeled trees. Our extension 
to \emph{labeled graphs} may be of independent interest, as it has potentially many more applications inside 
and outside computational biology.

The focus of this paper is the finite automaton 
representation of a multiple alignment. This setting is closely related to our 
previous work on indexing highly repetitive sequence collections 
\cite{Maekinen2010}. In our previous work, we represented a collection of 
individual genomes of total length $N$, with reference sequence of length $n$, and 
a total of $s$ mutations, in space $O(n\log \frac{N}{n} + s \log^2 N)$ bits in the 
average case (rough upper bounds here for simplicity). Exact pattern matching was 
supported in $O(m \log N)$ time. The index proposed in this paper achieves 
$O(n(1+s/n)^{O(\log n)})$ bits in the expected case for constant-sized alphabets.

The paper is organized as follows. 
Section~\ref{sect:defs} introduces the notation, 
Sect.~\ref{sect:indexes} reviews the necessary index structures we build on, 
Sect.~\ref{sect:fl} describes the new extension to finite languages,  
Sect.~\ref{sect:construction} shows how to construct the new index given a multiple alignment,
Sect.~\ref{sect:experiments} gives some preliminary 
experiments on SNP calling problem, and 
Sect.~\ref{sect:discussion} concludes the work by sketching the steps required 
for making the index into a fully applicable tool.

\section{Definitions}
\label{sect:defs}

A {\em string} $S = S[1,n]$ is a {\em sequence} of {\em characters} from {\em alphabet} $\Sigma = \{ 1, 2, \dotsc, \sigma \}$. A {\em substring} of $S$ is written as $S[i,j]$. A substring of type $S[1,j]$ is called a {\em prefix}, while a substring of type $S[i,n]$ is called a {\em suffix}. A {\em text} string $T = T[1,n]$ is a string terminated by $T[n] = \$ \not\in \Sigma$ with lexicographic value $0$. The {\em lexicographic order} ''$<$'' among strings is defined in the usual way.

A {\em graph} $G = (V, E)$ consists of a set $V = \set{v_{1}, \dotsc, 
v_{\abs{V}}}$ of {\em nodes} and a set $E \subset V^{2}$ of {\em edges} such that 
$(v, v) \not\in E$ for all $v \in V$. We call $(u, v) \in E$ an edge from node $u$ 
to node $v$. A graph is {\em directed}, if edge $(u, v)$ is distinct from edge 
$(v, u)$. In this paper, the graphs are always directed. For every node $v \in V$, 
the {\em indegree} $in(v)$ is the number of incoming edges $(u, v)$, and the {\em outdegree} $out(v)$ the number of outgoing edges $(v, w)$.

Graph $G$ is said to be {\em labeled}, if we attach a {\em label} $\ell(v)$ to 
each node $v \in V$. A {\em path} $P = u_{1} \dotsm u_{\abs{P}}$ is a sequence of 
nodes from $u_{1}$ to $u_{\abs{P}}$ such that $(u_{i}, u_{i+1}) \in E$ for all $i 
< \abs{P}$. The label of path $P$ is the string $\ell(P) = \ell(u_{1}) \dotsm 
\ell(u_{\abs{P}})$. A {\em cycle} is a path from a node to itself containing, at 
least one other node. If a graph contains no cycles, it is called {\em acyclic}.

A {\em finite automaton} is a directed labeled graph $A = (V, E)$.\footnote{Unlike the usual definition, we label nodes instead of edges.} The {\em initial node} $v_{1}$ is labeled with $\ell(v_{1}) = \#$ with lexicographic value $\sigma + 1$, while the {\em final node} $v_{\abs{V}}$ is labeled with $\ell(v_{\abs{V}}) = \$$. The rest of the nodes are labeled with characters from alphabet $\Sigma$. Every node is assumed to be on some path from $v_{1}$ to $v_{\abs{V}}$.

The {\em language} $L(A)$ \emph{recognized} by automaton $A$ is the set of the labels of 
all paths from $v_{1}$ to $v_{\abs{V}}$. We say that automaton $A$ recognizes a string $S \in L(A)$, and that a suffix $S'$ can be recognized from node $v$, if there is a path from $v$ to $v_{\abs{V}}$ with label $S'$.
Note that all strings in the language are of form $\# x \$$, where $x$ is a string from alphabet $\Sigma$.
If the language 
contains a finite number of strings, it is called {\em finite}. A language is 
finite if and only if the automaton is acyclic. Two automata are said to be {\em 
equivalent}, if they recognize the same language.

Automaton $A$ is forward (reverse) {\em deterministic} if, for every node $v \in V$ and every character $c \in \Sigma \cup \set{\#, \$}$, there exists at most one node $u$ such that $\ell(u) = c$ and $(v, u) \in E$ ($(u, v) \in E$). For any language recognized by some finite automaton, we can always construct an equivalent automaton that is forward (reverse) deterministic.

\section{Compressed indexes}
\label{sect:indexes}

The {\em suffix array (SA)} of text $T[1,n]$ is an array of pointers $\SA[1,n]$ to the suffixes of $T$ in lexicographic order. As an abstract data type, a suffix array is any data structure that supports the following operations efficiently: (a) {\em find} the SA range containing the suffixes prefixed by  \emph{pattern} $P$; (b) {\em locate} the suffix $\SA[i]$ in the text; and (c) {\em display} any substring of text $T$.
 
{\em Compressed suffix arrays (CSA)} \cite{Grossi2005,Ferragina2005a} 
support these operations. Their compression is based on the {\em Burrows-Wheeler 
transform (BWT)} \cite{Burrows1994}, a permutation of the text related to the SA. 
The BWT of text $T$ is a sequence $\BWT[1,n]$ such that $\BWT[i] = T[\SA[i] - 1]$, 
if $\SA[i] > 1$, and $\BWT[i] = T[n] = \$$ otherwise.
 
BWT can be reversed by a permutation called {\em $LF$\nobreakdash-mapping} \cite{Burrows1994,Ferragina2005a}. Let $C[1,\sigma]$ be an array such that $C[c]$ is the number of characters in $\{ \$, 1, 2, \dotsc, c-1 \}$ occurring in the BWT. We also define $C[0] = C[\$] = 0$ and $C[\sigma + 1] = n$. We define $LF$\nobreakdash-mapping as $LF(i) = C[\BWT[i]] + rank_{\BWT[i]}(\BWT,i)$, where $rank_{c}(\BWT,i)$ is the number of occurrences of character $c$ in prefix $\BWT[1,i]$.
 
The inverse of $LF$\nobreakdash-mapping is $\Psi(i) = select_{c}(\BWT, i - C[c])$, where $c$ is the highest value with $C[c] < i$, and $select_{c}(\BWT,j)$ is the position of the $j$th occurrence of character $c$ in $\BWT$ \cite{Grossi2005}. 
By its definition, function $\Psi$ is strictly increasing in the range $[C[c] + 1, C[c+1]]$ for every $c \in \Sigma$. Note that $T[\SA[i]] = c$ and $\BWT[\Psi(i)] = c$ for $C[c] < i \le C[c+1]$.
 
These functions form the backbone of CSAs. As $\SA[LF(i)] = \SA[i]-1$ \cite{Ferragina2005a} and hence $\SA[\Psi(i)] = \SA[i]+1$, we can use these functions to move the suffix array position backward and forward in the sequence. The functions can be efficiently implemented by adding some extra information to a compressed representation of the BWT. Standard techniques \cite{Navarro2007} for supporting SA functionality include using {\em backward searching} \cite{Ferragina2005a} for {\em find}, and sampling some suffix array values for {\em locate} and {\em display}.

{\em XBW} \cite{Ferragina2009b} is a generalization of the Burrows-Wheeler transform for labeled trees, where leaf nodes and internal nodes are labeled with different alphabets. Internal nodes of the tree are sorted lexicographically according to path labels from the node to the root. Sequence $\BWT$ is formed by concatenating the labels of the children of each internal node in lexicographic order according to the parent node. Every internal node $v$ now corresponds to a substring $\BWT[sp_{v}, ep_{v}]$ containing the labels of its children. The first position $sp_{v}$ of each such substring is marked with a \onebit\ in bit vector $F$. Backward searching is used to support the analogue of {\em find}. Tree navigation is possible by using $\BWT$ and $F$.

\section{Burrows-Wheeler transform for finite languages}
\label{sect:fl}

In this section, we generalize the XBW approach to finite automata. We call it the {\em generalized compressed suffix array (GCSA)}. For the GCSA to function, we require that the automaton is prefix-sorted. Refer to Section~\ref{sect:construction} on how to transform an automaton into an equivalent prefix-sorted automaton.

\begin{definition}
Let $A$ be a finite automaton, and let $v \in V$ be a node. Node $v$ is {\em prefix-sorted} by prefix $p(v)$, if the labels of all paths from $v$ to $v_{\abs{V}}$ share a common prefix $p(v)$, and no path from any other node $u \ne v$ to $v_{\abs{V}}$ has $p(v)$ as a prefix of its label. Automaton $A$ is prefix-sorted, if all nodes are prefix-sorted.
\end{definition}

Every node of a prefix-sorted automaton $A$ corresponds to a lexicographic range of suffixes of language $L(A)$. These ranges do not overlap for any two nodes.

In XBW, bit vector $F$ is used to mark both nodes and edges. If node $v$ has lexicographic rank $i$, the labels of its predecessors are $\BWT[sp_{v}, ep_{v}] = \BWT[select_{1}(F,i), select_{1}(F,i+1)-1]$. On the other hand, if node $u$ is a child of node $v$, and $\BWT[j]$ contains the label of node $u$, then $LF(j)$ is the lexicographic rank of the label of the path from node $u$ through node $v$ to the root. Hence $select_{1}(F, LF(j))$ gives us the position of edge $(u,v)$.

While the latter functionality is trivial in trees, a node can have many outgoing edges in a finite automaton. Hence we will use another bit vector $M$ to mark the outgoing edges.

Let $A = (V, E)$ be a prefix-sorted automaton. To build GCSA, we sort the nodes $v \in V$ according to prefixes $p(v)$. For every node $v \in V$, sequence $\BWT$ and bit vectors $F$ and $M$ contain range $[sp_{v}, ep_{v}]$ of length $n(v) = \max(in(v), out(v))$. See Figure~\ref{fig:sorted} and Table~\ref{table:gcsa} for an example.
\begin{itemize}

\item $\BWT[sp_{v}, ep_{v}]$ contains the labels $\ell(u)$ for all incoming edges $(u, v) \in E$, followed by $n(v) - in(v)$ empty characters.

\item $F[sp_{v}] = 1$ and $F[i] = 0$ for $sp_{v} < i \le ep_{v}$.

\item $M[sp_{v}, ep_{v}]$ contains $out(v)$ \onebit{}s followed by $n(v) - out(v)$ \zerobit{}s. For the final node $v_{\abs{V}}$, the range contains one $1$\nobreakdash-bit followed by $0$\nobreakdash-bits.

\end{itemize}

Array $C$ is used with some modifications. We define $C[\sigma + 1] = C[\#]$ in the same way as for regular characters, while $C[\sigma + 2]$ is set to be $\abs{E}$. Assuming that each edge $(u, v) \in E$ has an implicit label $\ell(u) p(v)$, we can interpret $C[c]$ as the number of edges with labels smaller than $c$. We write $char(i)$ to denote character $c$ such that $C[c] < i \le C[c+1]$.

\subsection{Basic navigation}

Let $[sp_{v}, ep_{v}]$ be the range of $\BWT$ corresponding to node $v \in V$. We define the following functions:
\begin{itemize}
\item $LF([sp_{v}, ep_{v}], c) = [sp_{u}, ep_{u}]$, where $\ell(u) = c$ and $(u, v) \in E$, or $\emptyset$ if no such $u$ exists.
\item $\Psi([sp_{u}, ep_{u}]) = \set{[sp_{v}, ep_{v}] \mid (u, v) \in E}$.
\item $\ell([sp_{v}, ep_{v}]) = \ell(v)$.
\end{itemize}
These are generalizations of the respective functions on BWT. $LF$ can be used to move backwards on edge $(u, v)$ such that $\ell(u) = c$, while $\Psi$ lists the endpoints of all outgoing edges from node $u$. These functions can be implemented by using $\BWT$, $F$, $M$, and $C$, as seen in Figure~\ref{fig:navigation}.

\begin{figure}[t!]\centering
\begin{minipage}[t]{.5\textwidth}
\begin{tabbing}
mm\=m\=m\= \kill
{\bf function} $LF([sp_{v}, ep_{v}], c)$: \\
1 \> $i \leftarrow C[c] + rank_{c}(\BWT, ep_{v})$ \\
2 \> {\bf if} $select_{c}(\BWT, i - C[c]) < sp_{v}$: \\
3 \> \> {\bf return} $\emptyset$ \\
4 \> $i \leftarrow select_{1}(M, i)$ \\
5 \> $sp_{u} \leftarrow select_{1}(F, rank_{1}(F, i))$ \\
6 \> $ep_{u} \leftarrow select_{1}(F, rank_{1}(F, i) + 1) - 1$ \\
7 \> {\bf return} $[sp_{u}, ep_{u}]$ \\
\\
{\bf function} $\ell([sp_{v}, ep_{v}])$: \\
8 \> {\bf return} $char(rank_{1}(M, sp_{v}))$ \\
\\
\end{tabbing}
\end{minipage}
\hspace{5pt}
\begin{minipage}[t]{.5\textwidth}
\begin{tabbing}
mm\=m\=m\= \kill
{\bf function} $\Psi([sp_{u}, ep_{u}])$: \\
9  \> $c \leftarrow \ell([sp_{u}, ep_{u}])$ \\
10 \> $res \leftarrow \emptyset$ \\
11 \> $low \leftarrow rank_{1}(M, sp_{u})$ \\
12 \> $high \leftarrow rank_{1}(M, ep_{u})$ \\
13 \> {\bf for} $i \leftarrow low$ {\bf to} $high$: \\
14 \> \> $j \leftarrow select_{c}(\BWT, i - C[c])$ \\
15 \> \> $sp_{v} \leftarrow select_{1}(F, rank_{1}(F, j))$ \\
16 \> \> $ep_{v} \leftarrow select_{1}(F, rank_{1}(F, j) + 1) - 1$ \\
17 \> \> $res \leftarrow res \cup \set{[sp_{v}, ep_{v}]}$ \\
18 \> {\bf return} $res$
\end{tabbing}
\end{minipage}

\caption{Pseudocode for the basic navigation functions $LF$, $\Psi$, and $\ell$.}\label{fig:navigation}
\end{figure}

Line 1 of $LF$ is similar to the regular $LF$, determining the rank of edge label 
$cp(v)$ among all edge labels. On lines 2 and 3, we determine if there is an 
occurrence of $c$ in $\BWT[sp_{v}, ep_{v}]$. On line 4, we find the position of 
edge $(u, v)$ in bit vector $M$, and on lines 5 and 6, we find the range $[sp_{u}, 
ep_{u}]$ containing this position.

In function $\Psi$, we determine the ranks of the outgoing edges from node $u$ on lines 11 and 12. Line 14 is similar to the regular $\Psi$, determining where the label of node $u$ corresponding to the edge of rank $i$ occurs in $\BWT$. Lines 15 and 16 are similar to lines 5 and 6 in $LF$, converting the position $j$ to the range $[sp_{v}, ep_{v}]$ corresponding to the destination node $v$.

\subsection{Searching}

As the generalized compressed suffix array is a CSA, most of the algorithms using a CSA can be modified to use GCSA instead. In this section, we describe how to support the basic SA operations (see Sect.~\ref{sect:indexes}):

\begin{itemize}
\item {\em find}($P$) returns the range $[sp, ep]$ of $\BWT$ corresponding to those nodes $v$, where at least one path starting from $v$ has pattern $P$ as a prefix of its label.
\item {\em locate}($[sp_{v}, ep_{v}]$) returns a numerical value corresponding to node $v$.
\item {\em display}($[sp_{u}, ep_{u}], k$) returns a prefix of the label of the path starting from node $u$. Stops when the prefix has length $k$ or when there are multiple or no outgoing edges from the current node.
\end{itemize}

We use backward searching \cite{Ferragina2005a} to support {\em find}. The algorithm maintains an invariant that $[sp_{i}, ep_{i}]$ is the range returned by {\em find}($P[i,m]$). In the initial step, we start with the edge range $[C[P[m]] + 1, C[P[m]+1]]$, and convert it to range $[sp_{m}, ep_{m}]$ by using bit vectors $M$ and $F$. The step from $[sp_{i+1}, ep_{i+1}]$ to $[sp_{i}, ep_{i}]$ is a generalization of function $LF$ for a range of nodes. We find the first and last occurrences 
of character $P[i]$ in $\BWT[sp_{i+1}, ep_{i+1}]$, map them to edge ranks by using $C$ and $\BWT$, and convert the ranks to $sp_{i}$ and $ep_{i}$ by using $F$ and $M$.

For {\em locate}, we assume that there is a (not necessarily unique) numerical value $id(v)$ attached to each node $v \in V$. Examples of these values include node ids (so that $id(v_{i}) = i$) and positions in the multiple alignment. To avoid excessive sampling of node values, $id(v)$ should be $id(u) + 1$ whenever $(u,v)$ is the only outgoing edge from $u$ and the only incoming edge to $v$.

We sample $id(u)$, if there are multiple or no outgoing edges from node $u$, or if $id(v) \ne id(u) + 1$ for the only outgoing edge $(u, v)$. We also sample one out of $d$ node values, given sample rate $d > 0$, on paths of at least $d$ nodes without any samples. The sampled values are stored in the same order as the nodes, and their positions are marked in bit vector $B$ ($B[sp_{u}] = 1$, if node $u$ is sampled).

Node values are retrieved in a similar way as in CSAs \cite{Navarro2007}. To retrieve $id(u)$, we first check if $B[sp_{u}] = 1$, and return sample $rank_{1}(B, sp_{u})$, if this is the case. Otherwise we follow the only outgoing edge $(u,v)$ by using function $\Psi$, and continue from node $v$. When we find a sampled node $w$, we return $id(w) - k$, where $k$ is the number of steps taken by using $\Psi$.

{\em Display} is implemented in a straightforward way. We first output $\ell([sp_{u}, ep_{u}])$, and then use $\Psi([sp_{u}, ep_{u}])$ to get the outgoing edges. If there are multiple or no edges, we stop. Otherwise $(u, v)$ is the only edge, and we continue from node $v$ until $k$ characters have been output.

A better analogue to the {\em display} of suffix arrays would be {\em display}($id(u), k$), extracting prefixes of path labels from all nodes $v \in V$ with node value $id(u)$. To implement such operation, we need am efficient way to map node values to $\BWT$ ranges. While some node value schemes allow such mapping easily, there is no obvious way to do so in other schemes.

\subsection{Analysis}

For each node $v \in V$, the length of range $[sp_{v}, ep_{v}]$ is the maximum of $in(v)$ and $out(v)$. As every node must have at least one incoming edge and one outgoing edge (except for the initial and the final nodes), the length of $\BWT$ is at most $2\abs{E} - \abs{V} + 2$.

Similar size bounds as for different variants of the CSA can be defined for compressed representations of $\BWT$, if we first define a generalization of empirical entropy. Bit vectors $F$ and $M$ have $\abs{V}$ and $\abs{E}$ $1$\nobreakdash-bits out of $\abs{\BWT}$, respectively. The number and the size of the samples depend greatly on the node value scheme used.

\begin{theorem}
Assume that $rank$ and $select$ on bit vectors require $O(t_{B})$ time. GCSA with sample rate $d$ supports {\em find}($P$) in $O(\abs{P} \cdot t_{B})$ and {\em locate}($[sp_{v}, ep_{v}]$) in $O(d \cdot t_{B})$ time.
\end{theorem}

\begin{proof}
We use bit vectors $\Psi_{c}$ that mark the occurrences of character $c \in \Sigma \cup \set{\#}$ to encode $\BWT$. This reduces {\em rank} and {\em select} on $\BWT$ to the same operations on bit vectors. Basic operations $\ell$ and $LF$ take $O(t_{B})$ time, as they require a constant number of bit vector operations. $\Psi$ also takes $O(t_{B})$ time, if the current node has outdegree $1$. As {\em find} does one generalized $LF$ per character of pattern, it takes $O(\abs{P} \cdot t_{B})$ time.

Operation {\em locate} checks from bit vector $B$ if the current position is sampled, and follows the unique outgoing edge using $\Psi$ if not. This requires a constant number of bit vector operations per step. As a sample is found within $d-1$ steps, the time complexity is $O(d \cdot t_{B})$. \qed
\end{proof}

\subsection{Languages recognized by a prefix-sorted automaton}
\label{sect:languages}

GCSA can be extended from finite languages to some infinite languages as well. To define the class of languages that can be indexed with our approach, we relax the requirement that the automaton should be prefix-sorted.

\begin{definition}
Let $A = (V, E)$ be a finite automaton, and let $v \in V$ be a node. Let $rng(v)$ be the smallest (open, semiopen, or closed) lexicographic range containing all suffixes that can be recognized from node $v$. Node $v$ is {\em prefix-range-sorted}, if no suffix $S \in rng(v)$ is recognized from any other node $v' \ne v$. Automaton $A$ is prefix-range-sorted, if all nodes are prefix-range-sorted.
\end{definition}

The definition states that the ranges of suffixes recognized from two nodes must not overlap. When this is true, the incoming edge encoded by character $c$ of rank $i$ in $\BWT$ is the same edge as the outgoing edge encoded by the \onebit{} of rank $C[c] + i$ in bit vector $M$.

\begin{theorem}
The class of languages recognized by prefix-range-sorted automata is strictly between finite languages and regular languages.
\end{theorem}

\begin{proof}
Consider the regular infinite language $\set{ \# x \$ \mid x \in \set{a, b}^{\ast}}$. The minimal automaton recognizing this language is prefix-range-sorted, as each node has a distinct label.

Now consider the regular language $L = \set{ \# x \$ \mid x \in \set{a, b}^{\ast} \cup \set{a, c}^{\ast}}$. Assume that there is a prefix-range-sorted automaton recognizing language $L$. Suffixes $B_{n} = a^{n} b \$$ and $C_{n} = a^{n} c \$$ must be recognized from different nodes, as $bB_{n}$ is a suffix of language $L$, while $bC_{n}$ is not. Suffixes $B_{n}$ and $B_{n+1}$ must also be recognized from different nodes, as $B_{n+1} < C_{n+1} < B_{n}$. As the automaton must have an infinite number of nodes, it cannot be prefix-range-sorted. \qed
\end{proof}

\section{Index construction}
\label{sect:construction}

Our construction algorithm is related to the {\em prefix-doubling} approach to suffix array construction \cite{Puglisi2007}. We start with a reverse deterministic automaton (see Figure~\ref{fig:automaton}), convert it to an equivalent prefix-range-sorted automaton (Figure~\ref{fig:sorted}), and build the GCSA (Table~\ref{table:gcsa}) for that automaton. The algorithm consists mostly of sorting, scanning, and database joins. Hence it can be efficiently implemented in parallel, distributed, and external memory settings.

\begin{theorem}
Assume we have a length $n$ multiple alignment of $r$ sequences from alphabet of size $\sigma$. We can build a prefix-range-sorted automaton recognizing all paths through the alignment in $O(nr \log r + \abs{V'} \log \abs{V'} \log n + \abs{E'})$ time and $O(nr \log \sigma + \abs{V'} \log \abs{V'} + \abs{E'} \log \abs{E'})$ bits of space, where $V'$ and $E'$ are the largest intermediate sets of nodes and edges, respectively.
\end{theorem}

\begin{proof}
From Lemmas \ref{lemma:automaton construction}, \ref{lemma:node construction}, and \ref{lemma:edge construction} below. \qed
\end{proof}

The sizes of the largest intermediate sets of nodes and edges are analyzed in a restricted model in the Appendix.

\begin{figure}[t!]
\centerline{\includegraphics{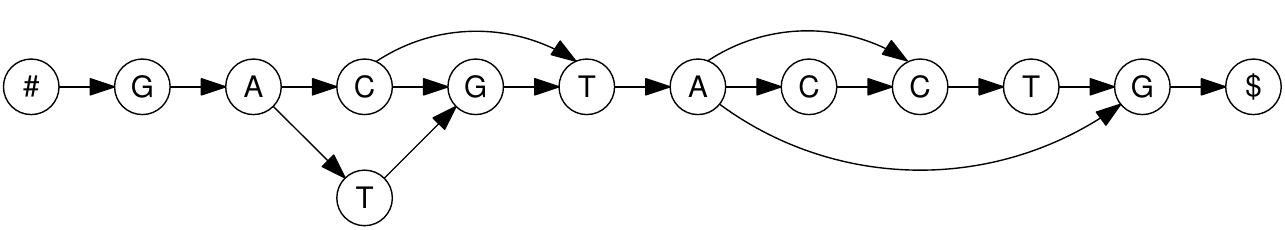}}
\caption{A reverse deterministic automaton corresponding to the first 10 positions of the multiple alignment in Figure~\ref{fig:multiplealignment}.} 
\label{fig:automaton}
\end{figure}

\begin{figure}[t!]
\centerline{\includegraphics{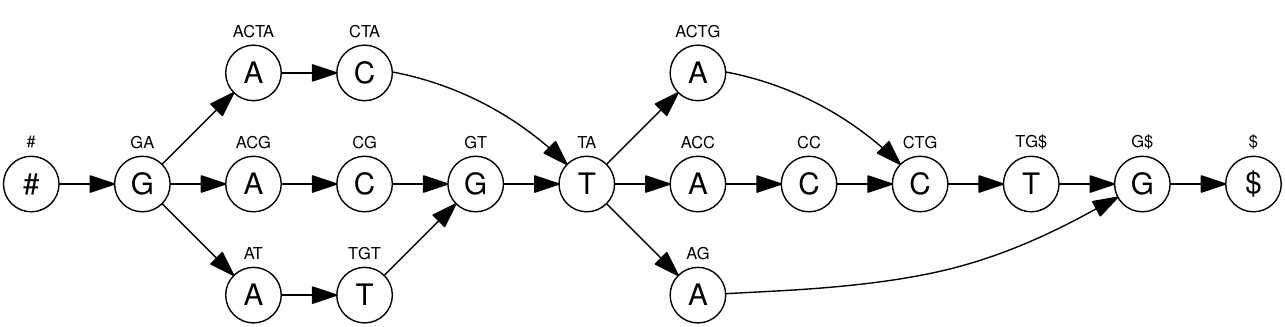}}
\caption{A prefix-sorted automaton built for the automaton in Figure~\ref{fig:automaton}. The strings above nodes are prefixes $p(v)$.} 
\label{fig:sorted}
\end{figure}

\begin{table}[t!]
\centering
\caption{GCSA for the automaton in Figure~\ref{fig:sorted}. Nodes are identified by prefixes $p(v)$.} 
\label{table:gcsa}
\renewcommand{\tabcolsep}{.05cm}
\texttt{
\begin{tabular}{cccccccccccccccccccc}
\hline\noalign{\smallskip}
       & & \$  & ACC & ACG & ACTA & ACTG & AG & AT & CC & CG & CTA & CTG & G\$ & GA   & GT & TA  & TG\$ & TGT & \# \\
\noalign{\smallskip}
\hline
\noalign{\smallskip}
$\BWT$ & & G   & T   & G   & G    & T    & T  & G  & A  & A  & A   & AC  & AT  & \#-- & CT & CG- & C    & A   & \$ \\
$F$    & & 1   & 1   & 1   & 1    & 1    & 1  & 1  & 1  & 1  & 1   & 10  & 10  & 100  & 10 & 100 & 1    & 1   & 1 \\
$M$    & & 1   & 1   & 1   & 1    & 1    & 1  & 1  & 1  & 1  & 1   & 10  & 10  & 111  & 10 & 111 & 1    & 1   & 1 \\
\noalign{\smallskip}
\hline
\end{tabular}}
\end{table}

\subsection{Building the automaton}

In this section, we show how to construct a reverse deterministic automaton from a multiple alignment of sequences. With a similar approach, we can also construct the automaton from a reference sequence and a set of SNPs.

To efficiently build a reverse deterministic automaton, we have to modify the sequences. Let $S_{1}, \dotsc, S_{r}$ be sequences of length $n$ from alphabet $\Sigma \cup \set{-}$, where $-$ represents a gap in the sequence. Consider position $j$. For every sequence $S_{i}$, let $S_{i}[j-k_{i}] = c_{i}$ be the first non-gap character preceding $S_{i}[j]$. We say that sequences $S_{i}$ and $S_{i'}$ are equivalent at position $j$, if $S_{i}[j] = S_{i'}[j] \ne -$ and $c_{i} = c_{i'}$. If sequences $S_{i}$ and $S_{i'}$ are equivalent at position $j$, we move the preceding characters to position $\max(j-k_{i}, j-k_{i'})$.

These modifications can be done in one pass over the sequences in reverse direction. Then, if there are positions with a gap in every sequence, we remove these positions from the sequences. Finally, we form new sequences $T_{i} = \# S_{i} \$$ for all $i$.

\begin{claim}
Let $j$ be a position such that $T_{i}[j] = T_{i'}[j]$ for some $i \ne i'$, and let $T_{i}[j-k]$ and $T_{i'}[j-k']$ be the preceding non-gap characters. If $T_{i}[j-k] = T_{i'}[j-k']$, then $k = k'$.
\end{claim}

Let $m \ge 0$ be the desired context length.\footnote{Two mutations within $m$ 
positions in the same sequence are considered to be parts of the same mutation. 
Increasing context length generally decreases index size, construction time, and 
construction space.} For every sequence $i$ and position $j$, where $T_{i}[j] \ne 
-$, we define a label $\ell(i, j)$ of length $m+1$ consisting of the character 
$T_{i}[j]$ and the next $m$ non-gap characters (context). If there are not enough 
characters left, we add duplicates of $\$$ to the end of the label.

For every non-gap character $T_{i}[j]$, we create a temporary node $v_{i,j}$ with label $\ell(v_{i,j}) = T_{i}[j]$ and an edge $(v_{i,j-k}, v_{i,j})$, where $T_{i}[j-k]$ is the non-gap character preceding $T_{i}[j]$. Then, for all positions $j > 1$ and all sequences $T_{i} \ne T_{i'}$, we merge nodes $v_{i,j}$ and $v_{i',j}$, if $\ell(i,j) = 
\ell(i',j)$, to get the actual nodes. For position $j = 1$, we merge nodes $v_{i, 
1}$ for all sequences $i$.

\begin{claim}
Let $(v_{i,j-k}, v_{i,j})$ and $(v_{i',j-k'}, v_{i,j})$ be two edges with 
$\ell(v_{i,j-k}) = \ell(v_{i',j-k'})$. Then $k = k'$, $\ell(i,j-k) = 
\ell(i',j-k') = \ell(v_{i,j-k})\ell(i,j)[1,m]$, and hence nodes $v_{i,j-k}$ and 
$v_{i',j-k'}$ will be merged, making the automaton reverse deterministic.
\end{claim}

\begin{lemma}\label{lemma:automaton construction}
Let $n$ be the length of the multiple alignment, $r$ the number of sequences, $\sigma$ the size of the alphabet, and $m$ the context length. Building a reverse deterministic automaton $A = (V, E)$ takes $O(nr \log r)$ time and requires $O(nr \log \sigma + \abs{E} \log \abs{E})$ bits of space.
\end{lemma}

\begin{proof}
The alignment modifications can be done in $O(r \log r)$ time per position by sorting the pairs $(c_{i}, k_{i})$ and checking, if there are multiple values of $k$ per character. To create the nodes, we have to sort the labels $\ell(i,j)$ for each position $j$. If we scan the alignment in reverse direction, we can maintain this order in $O(r)$ time per position. Creating the edges also takes $O(r)$ time per position, as at most $r$ edges are created. Space requirements come from storing the alignment and the automaton. \qed
\end{proof}

\subsection{Creating the nodes of a prefix-sorted automaton}
\label{sect:prefix-sorting}

\begin{definition}
Let $A$ be a finite automaton recognizing a finite language, and let $k > 0$ be an integer. Automaton $A$ is {\em $k$\nobreakdash-sorted} if, for every node $v \in V$, the labels of all paths from $v$ to $v_{\abs{V}}$ share a common prefix $p(v, k)$ of length $k$, or if node $v$ is prefix-sorted by prefix $p(v, k)$ of length at most $k$.
\end{definition}

Note that every automaton is $1$\nobreakdash-sorted. Automaton $A$ is prefix-sorted if and only if it is $n$\nobreakdash-sorted, where $n$ is the length of the longest string in $L(A)$.

Starting from a reverse deterministic automaton $A = A_{0}$, we create the nodes of automata $A_{i} = (V_{i}, E_{i})$ for $i = 1, 2, \dotsc$ that are $2^{i}$\nobreakdash-sorted, until we get an automaton that is prefix-sorted. For every node $v \in V_{i}$, let $P(v)$ be the path of $A$ corresponding to prefix $p(v, 2^{i})$. We store the first and the last nodes of this path as $\from(v)$ and $to(v)$, and set $rank(v)$ to be the lexicographic rank of prefix $p(v, 2^{i})$ among all distinct prefixes $p(u, 2^{i})$ of nodes $u \in V_{i}$. Value $sorted(v)$ is used to indicate whether the node is prefix-sorted.

\begin{claim}
Node $v$ is prefix-sorted if and only if $rank(v)$ is unique.
\end{claim}

The basic step of the algorithm is the {\em doubling} step from $A_{i}$ to $A_{i+1}$. If node $u \in V_{i}$ is prefix-sorted, we {\em duplicate} it as $w \in V_{i+1}$, and set $rank(w) = (rank(u), 0)$. Otherwise we create a {\em joined} node $uv \in V_{i+1}$ for every node $v \in V_{i}$ such that $P(uv) = P(u)P(v)$ is a path in $A$, and set $\ell(uv) = \ell(u)$ and $rank(uv) = (rank(u), rank(v))$. As path $P(uv)$ exists if and only if there is an edge $(to(u), \from(v)) \in E_{0}$, this essentially requires two database joins. When the nodes of $A_{i+1}$ have been created, we sort them by their ranks, and replace the pairs of integers with integer ranks.

The doubling step is followed by the {\em pruning} step, where we merge redundant nodes. The nodes in $V_{i+1}$ are sorted by their $rank(\cdot)$ values. If all nodes sharing a certain $rank(\cdot)$ value also share their $\from(\cdot)$ node, these nodes are equivalent, and can be merged.

\begin{claim}
After doubling and pruning steps, automaton $A_{i+1}$ is $2^{i+1}$\nobreakdash-sorted, and recognizes language $L(A_{i})$.
\end{claim}

\begin{lemma}\label{lemma:node construction}
Creating the nodes of prefix-sorted automaton $A_{i}$ takes $O(\abs{V'} \log \abs{V'} \log n)$ time and requires $O(\abs{V'} \log \abs{V'})$ bits of space in addition to automaton $A$, where $V'$ is the largest set of nodes during construction, and $n$ is the length of the longest string in $L(A)$.
\end{lemma}

\begin{proof}
We can implement the doubling and pruning steps by scanning and sorting the nodes several times. Hence each step requires at most $O(\abs{V'} \log \abs{V'})$ time. As we need at most $\lceil \log n \rceil$ doubling steps to get a prefix-sorted automaton, the time bound follows. The space bound of $O(\abs{V'} \log \abs{V'})$ bits is the space required to store $\abs{V'}$ nodes. \qed
\end{proof}

\subsection{Creating the edges}

Let $A = (V, E)$ be a reverse deterministic automaton recognizing a finite language, and let $W$ be the set of nodes of an equivalent prefix-sorted automaton. To create the edges, we first merge nodes with adjacent $rank(\cdot)$ values, if they share their $\from(\cdot)$ node. The resulting set $V'$ is the set of nodes of a prefix-range-sorted automaton $A' = (V', E')$ equivalent to automaton $A$. The set of edges $E'$ can be constructed efficiently from automaton $A$ and the set of nodes $V'$.

\begin{claim}
For every node $v \in V'$, we have $\set{ \from(u) \mid (u, v) \in E' } = \set{ u \mid (u, \from(v)) \in E }$. Furthermore, there are no edges $(u, v), (u', v) \in E'$ such that $\from(u) = \from(u')$.
\end{claim}

We generate the incoming edges $(u, v) \in E'$ as a list of pairs $(\from(u), v)$, sorted by $(\ell(\from(u)), rank(v))$.

\begin{claim}
Edges sorted by $(\ell(\from(u)), rank(v))$ are also sorted by $rank(u)$.
\end{claim}

We replace each pair $(\from(u), v)$ with edge $(u, v)$ by scanning the sorted lists of nodes and edges. As every node (except the final node) has at least one outgoing edge, and no adjacent nodes share their $\from(\cdot)$ value, all adjacent edges with the same $\from(\cdot)$ value start from the current node. When the $\from(\cdot)$ value changes, we advance to the next node in the list.

\begin{lemma}\label{lemma:edge construction}
Creating the edge of prefix-range-sorted automaton $A'$ takes $O(\abs{W} + \abs{E'})$ time and requires $O(\abs{W} \log \abs{W} + \abs{E'} \log \abs{E'})$ bits of space.
\end{lemma}

\begin{proof}
Assuming that the set of nodes $W$ is already sorted, merging adjacent nodes takes $O(\abs{W})$ time. Creating the sorted list of pairs $(\from(u), v)$ takes $O(\abs{E'})$ time, as we can output the pairs into $\sigma$ buckets according to $\ell(\from(u))$, and the nodes are already sorted by $rank(v)$. Replacing the pairs with edges takes $O(\abs{E'})$ time. Space complexity comes from storing $A$, $W$, and $E'$. \qed
\end{proof}

\section{Implementation and Experiments}\label{sect:experiments}

We have implemented GCSA in C++, using the components from on our implementation of RLCSA \cite{Maekinen2010}.\footnote{\url{http://www.cs.helsinki.fi/group/suds/gcsa/}}
For each character $c \in \Sigma \cup \set{\#}$, we use a gap encoded bit vector to mark the occurrences 
of $c$ in $\BWT$. Bit vectors $F$ and $M$ are run-length encoded, as they usually consist of long runs of \onebit{}s. 
Bit vector $B$ is gap encoded, while the samples 
are stored using $\lceil \log (id_{\max} + 1) \rceil$ bits each, where $id_{\max}$ is the largest 
sampled value. Block size was set to 32 bytes in all bit vectors.

The implementation was compiled on g++ version 4.3.3. We used a system with 32 gigabytes of memory and two quad-core 2.53 GHz Intel Xeon E5540 processors running Ubuntu 10.04 with Linux kernel 2.6.32 for our experiments. Only one core was used in all experiments, except that GCSA construction used the parallel sorting algorithm provided by the \emph{libstdc++ parallel mode}.

\begin{figure}[t!]
\begin{center}
\includegraphics[width=\textwidth]{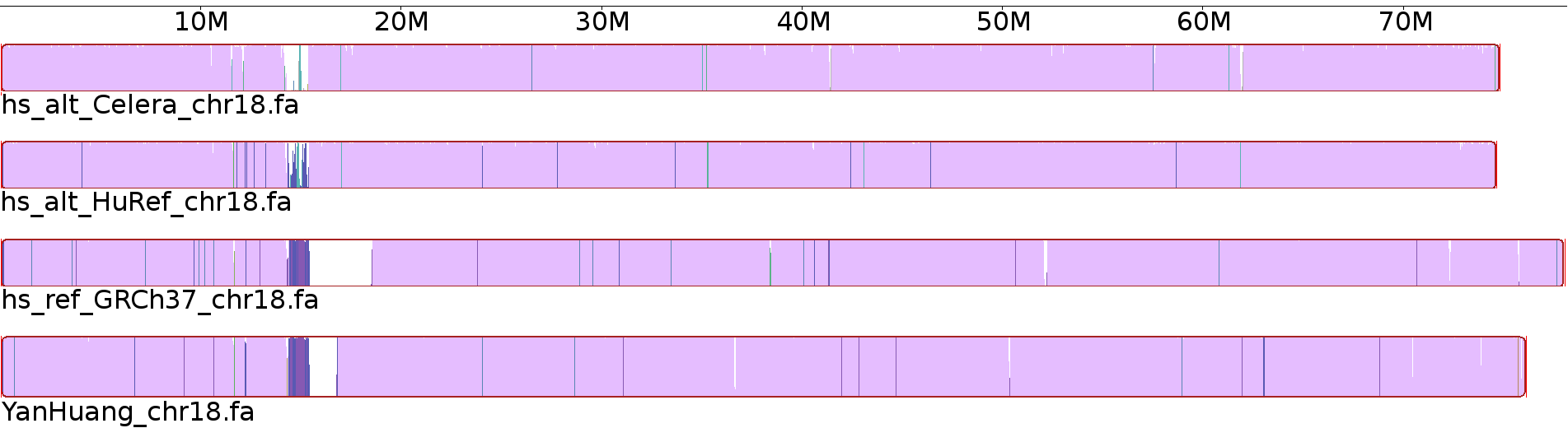}
\caption{Mauve output for the four different versions of chromosome 18.
Regions conserved in all versions are drawn with mauve color (the most prominent color).
The height of the color profile corresponds to the average level of conservation.
White areas could not be aligned and are probably specific to particular genome (or long runs of N).} 
\label{fig:malignment}
\end{center}
\end{figure}

We built a multiple alignment for 
four different assemblies of the human chromosome 18 (about 76 Mbp each).
Three of the assemblies were from 
NCBI\footnote{\url{ftp://ftp.ncbi.nih.gov/genomes/H_sapiens/Assembled_chromosomes/}}:
the assemblies by the Genome Reference Consortium ({\em GRCh37}), Celera Genomics ({\em Celera}), and J. Craig Venter Institute ({\em HuRef}). The fourth sequence\footnote{\url{ftp://public.genomics.org.cn/BGI/yanhuang/fa/}} was from
Beijing Genomics Institute ({\em YanHuang}). The sequences were aligned by the Mauve Multiple Genome Alignment software \cite{Darling2010}.
We ran progressiveMauve 2.3.1 assuming collinear genomes, as we do not support rearrangements.
The multiple alignment took a few hours to build: 
about 89.4\% of nucleotides aligned perfectly, 0.19\% with one or more mismatches, and 10.4\% were inside of a gap.
The number of gaps was high mainly because of the differences in the centromere region.
Fig.~\ref{fig:malignment} gives a rough visualization of the similarity between sequences.


For our experiments, we constructed a GCSA (sample rate $16$) for the alignment with various context lengths $m$. We also built RLCSA (sample rate $32$) and BWA 0.5.8a \cite{LD09} for the four sequences.
We searched for exact matches of 10 million Illumina/Solexa reads of length 56, sequenced from the whole genome, as both regular patterns and reverse complements. Table~\ref{table:construction} lists the results of these experiments. As there were relatively few occurrences inside the selected chromosome, most of the time was spent doing \emph{find}. Hence the sample rate that only affects \emph{locate} had little effect on the overall performance. 
GCSA was 2.5\nobreakdash--3 times slower than RLCSA and 3.5\nobreakdash--4 times slower than BWA. About 1\% of the reads matched by GCSA were not matched by the other indexes. Construction requirements for GCSA were higher than for the other indexes (see Sect.~\ref{sect:discussion} for discussion).

The performance gap between GCSA and RLCSA reflects differences in fundamental techniques, as the implementations share most of their basic components and design choices. Theoretically GCSA should be about 3 times slower, as it requires six bit vector operations per base in \emph{find}, while RLCSA uses just two. The differences between RLCSA and BWA come from implementation choices, as RLCSA is intended for highly repetitive sequences and BWA for fast pattern matching with DNA sequences.

\begin{table}[t!]
\centering
\caption{Index construction and exact matching with GCSA (sample rate $16$), RLCSA (sample rate $32$), and BWA on a multiple alignment of four sequences of human chromosome 18. Times for \emph{locate} include the time used by \emph{find}. GCSA\nobreakdash-$m$ denotes GCSA with context length $m$.}
\label{table:construction}
\renewcommand{\tabcolsep}{.1cm}
\begin{tabular}{lcrcrrcrrr}
\hline\noalign{\smallskip}
 & & & & \multicolumn{2}{c}{{\bf Construction}} & & \multicolumn{3}{c}{{\bf Matching}} \\
{\bf Index} & & {\bf Size} & & {\bf Time} & {\bf Space} & & {\bf Matches} & {\bf Find} & {\bf Locate} \\
\noalign{\smallskip}
\hline
\noalign{\smallskip}
GCSA-2 & &  69.3 MB & & 10 min & 7.0 GB & & 388,873 & 16 min & 20 min \\
GCSA-4 & &  67.2 MB & & 10 min & 6.7 GB & & 388,212 & 16 min & 18 min \\
GCSA-8 & &  64.7 MB & &  9 min & 4.8 GB & & 387,707 & 16 min & 18 min \\
RLCSA  & & 165.0 MB & &  5 min & 2.3 GB & & 384,400 &  6 min &  7 min \\
BWA    & & 212.4 MB & &  4 min & 1.4 GB & & 384,400 &      - &  5 min \\
\noalign{\smallskip}
\hline
\end{tabular}
\end{table}

To test GCSA in a more realistic alignment algorithm, we implemented BWA-like approximate searching \cite{LD09} for both GCSA and RLCSA. There are some differences to BWA: i) we return all best matches; ii) we do not use a seed sequence; iii) we have no limits on gaps; and iv) we have to match $O(|P| \log |P|)$ instead of $O(|P|)$ characters to build the lower bound array for pattern $P$, as we have not indexed the reverse sequence. We used context length $4$ for GCSA, as it had the best trade-off in exact matching.

The results can be seen in Table~\ref{table:comparison}. GCSA was about 2.5 times slower than RLCSA, while finding from 1.0\% (exact matching) to 2.4\% (edit distance $3$) more matches in addition to those found by RLCSA. BWA is significantly faster (e.g. finding 1,109,668 matches with $k=3$ in 40 minutes), as it solves a slighly different problem, ignoring a large part of the search space with biologically implausible edit operations. A fair comparison with BWA is currently impossible without significant amount of reverse engineering. With the same algorithm in all three indexes, the performance differences should be similar as in exact matching.

\begin{table}[t!]
\centering
\caption{Approximate matching with GCSA and RLCSA. The reported matches for given edit distance $k$ include those found with smaller edit distances.}
\label{table:comparison}
\renewcommand{\tabcolsep}{1mm}
{\begin{tabular}{lcrrcrr}
\hline\noalign{\smallskip}
 & & \multicolumn{2}{c}{{\bf GCSA-4}} & & \multicolumn{2}{c}{{\bf RLCSA}}  \\
$\mathbf{k}$ & & {\bf Matches} & {\bf Time} & & {\bf Matches} & {\bf Time} \\
\noalign{\smallskip}
\hline
\noalign{\smallskip}
$0$ & &   388,212 &    18 min & &   384,400 &   7 min  \\
$1$ & &   620,263 &   101 min & &   609,320 &  39 min  \\
$2$ & &   876,228 &   283 min & &   856,373 & 111 min  \\
$3$ & & 1,146,032 & 1,730 min & & 1,118,719 & 721 min  \\
\noalign{\smallskip}
\hline
\end{tabular}}
\end{table}

We also tried to configure BWA to use a more similar algorithm to ours. With no limitations on gaps, BWA found matches for 1,156,013 reads in 257 minutes with edit distance $3$, while the actual edit distance grew past $3$ in some cases. As the exact mechanism BWA uses to handle gaps is unknown, we could not implement it for GCSA and RLCSA.

Finally, we made a preliminary experiment on the SNP calling application mentioned in Sect.~\ref{sect:introduction} using in-house software. 
We called for SNPs from chromosome 18 with minimum coverage 2, using all 10 million reads, as well as only those reads with \emph{no} exact matches on GCSA-4.
The number of called SNPs was 4203 with all reads and 1074 with non-matching ones.
We did not yet compare how much of the reduction can be explained by exact matches on recombinants that would also be found using approximate search on one reference, and how much by more accurate alignment due to richer reference set.

\section{Discussion}
\label{sect:discussion}

Based on our experiments, GCSA is 2.5\nobreakdash--3 times slower than a similar implementation of CSA used in the same algorithm. With typical mutation rates, the index is also not much larger than a CSA built just for the reference sequence. Hence GCSA does not require significantly more resources than a regular compressed suffix array, while providing biologically relevant extended functionality.

While our construction algorithm uses more resources than CSA construction, genomes of up to about 100 Mbp can be indexed on a single workstation in a reasonable time. An external memory implementation should allow us to build an index for the human genome and all known SNPs in less than a day. Extrapolating from current results, the final index should be 2.5--3 gigabytes in size. As a faster alternative for indexing large genomes, we are also working on a distributed construction algorithm in the MapReduce framework \cite{Dean2004}. 

To improve the running time of short read alignment and related tasks, most of the search space pruning 
mechanisms (in addition to the one mechanism we already used from \cite{LD09}) to support approximate matching on top of BWT \cite{LTPS09,LD09,Lietal09,MVLK10} can be 
easily plugged in. Local alignment \cite{Lametal08,LD10} can also be supported. 

As mentioned in Section~\ref{sect:introduction}, an obvious generalization is to index labeled weighted graphs, 
where the weights correspond to probabilities for jumping from one sequence to another in the alignment. This does not increase space usage significantly, as the probabilities differ from $1.0$ only in nodes with multiple outgoing edges. 
In the restricted model analyzed in the Appendix, the extra space requirement is $O(pn \log n)$ bits for a reference sequence of length $n$ and mutation rate $p$.
During the construction of the index, it is also easy to discard paths with small probabilities, given a threshold. This approach can be used e.g. to index recombinants only in the recombination hotspot areas \cite{Mye05}.

The experiments conducted here aimed at demonstrating the feasibility and potential of the approach. 
Once we have the genome-scale implementation ready, we are able to test our 
claims on improving variant calling and primer design accuracy with the index. Both require 
wet-lab verification to see the true effect on reducing false positives.

\section*{Acknowledgements}

We wish to thank Eric Rivals for pointing out recombination hotspots and 
Riku Katainen for running the variation calling experiment.

\bibliographystyle{plain}
\bibliography{paper}

\newpage
\appendix
\section*{Appendix: Expected case analysis}
\label{appendix:size}

In the following, all random choices are independent and identically distributed (IID).

We analyze the size of the automata created by the doubling algorithm in the following model. 
Let $S[1,n]$ be a reference sequence, and let $p$ be the mutation rate. For each 
position $i = 1, \dots, n$, the initial automaton $A$ has a node $u_{i}$ with 
label $\ell(u_{i}) = S[i]$, randomly chosen from alphabet $\Sigma$. With 
probability $p$, there is also another node $w_{i}$ with a random label 
$\ell(w_{i}) \in \Sigma \setminus \set{S[i]}$. The automaton has edges from all nodes at position $i$ to all nodes at position $i+1$ for all $i$.

\begin{definition}
Let $k > 0$ be an integer. A {\em $k$-path} in an automaton is a path of length $k$, or a shorter path ending at the final node.
\end{definition}

Let $k > 0$ be an integer. For any position $i$, let $X_{i,k}$ be the number $k$-paths starting from position $i$ in the reference sequence. If there are $j$ mutated positions covered by these paths, then $X_{i,k} = 2^{j}$, and each of the paths has a different label. The number of mutations is binomially distributed, with the path length and the mutation probability as the parameters. From the moment-generating function for binomial distribution, we get
\begin{equation}\label{eq:expectation}
\Exp{X_{i,k}} = \sum_{j=0}^{k} \Pr(X_{i,k} = j) 2^{j} \le (1+p)^{k}.
\end{equation}
For positions $i = 1, \dots, n - k + 1$, this is an equality.

\begin{lemma}\label{lemma:nodes}
Let $A_{h}$ be a $2^{h}$-sorted automaton equivalent to the original automaton $A$. Then $N(2^{h}) = n(1+p)^{2^{h}} + 2$ is an upper bound for the expected number of nodes in $A_{h}$.
\end{lemma}

\begin{proof}
For every $2^{h}$-path starting from a position $i$ in the reference sequence, there is at most one node in automaton $A_{h}$. On the other hand, every node in the automaton, except for the initial and the final nodes, corresponds to a path that can be extended to some such $2^{h}$-path. Hence the total number of nodes is at most $\sum_{i=1}^{n} X_{i,k} + 2$. By Eq.~\ref{eq:expectation}, the expected number of nodes is at most $N(2^{h})$. \qed
\end{proof}

\begin{lemma}\label{lemma:edges}
Let $A_{h}$ be the $2^{h}$-sorted automaton built from automaton $A$. Then $N(2^{h})(1+p)$ is an upper bound for the expected number of edges in $A_{h}$.
\end{lemma}

\begin{proof}
The indegree of the initial node of $A_{h}$ is $0$. For every other node $v$, let $pos(v)$ be the position of the reference sequence corresponding to $\from(v)$. If $\from(v)$ is the final node of $A$, then $pos(v) = n+1$. If there is no mutation at position $pos(v) - 1$, then $in(v) = 1$. Otherwise $in(v) = 2$. Hence the expected number of edges is at most $(1+p)$ times the number of nodes. \qed
\end{proof}

Consider the expectation $\Exp{X_{i,k} X_{i',k}}$ for a pair of text positions $i < i'$. If $i' \ge i + k$, then the random variables are independent, and the expectation becomes
\begin{equation}\label{eq:independent}
\Exp{X_{i,k} X_{i',k}} = \Exp{X_{i,k}} \Exp{X_{i',k}} \le (1+p)^{2k}.
\end{equation}
Otherwise assume that the paths starting from positions $i$ and $i'$ overlap in $k' < k$ positions. Then the expectation is a product of the expectations of three independent random variables $X_{i,k-k'}$, $X_{i',k'}^{2}$, and $X_{i'+k',k-k'}$. By using the moment-generating function, we get
\begin{equation}\label{eq:dependent}
\Exp{X_{i,k} X_{i',k}}  \le (1+p)^{2(k-k')} (1+3p)^{k'} \le (1+p)^{3k}.
\end{equation}

\begin{definition}
A pair of nodes of automaton $A_{h}$ {\em collides}, if the corresponding $2^{h}$-paths have identical labels.
\end{definition}

\begin{lemma}\label{lemma:expectation}
Let $A_{h}$ be the $2^{h}$-sorted automaton built from automaton $A$ by using the 
doubling algorithm. The expected number of colliding pairs of nodes in automaton 
$A_{h}$ is at most $C(2^{h}) = n^{2} (1+p)^{3 \cdot 2^{h}} / \sigma^{2^{h}}$.
\end{lemma}

\begin{proof}
If two paths start from the same position in the reference sequence, the corresponding nodes cannot collide. As the colliding paths must be of length $2^{h}$ (otherwise the nodes would be prefix-sorted), the probability of collision of any given pair is $\sigma^{-2^{h}}$. By Equations \ref{eq:independent} and \ref{eq:dependent}, the expected number of colliding pairs is at most
\begin{displaymath}
\sum_{i < i'} \Exp{X_{i,2^{h}} X_{i',2^{h}} / \sigma^{2^{h}}} \le  n^{2} (1+p)^{3 \cdot 2^{h}} / \sigma^{2^{h}}.
\end{displaymath}
The lemma follows. \qed
\end{proof}

\begin{theorem}\label{theorem:nodes}
Let $n$ be the length of the reference sequence, $\sigma$ the size of the alphabet, and $p < \sigma^{1/3} - 1$ the mutation rate. For any $\varepsilon > 0$, the largest automaton created by the doubling algorithm has at most $n(1+p)^{k} + 2$ nodes with probability $1-\varepsilon$, where $k = 2 \log_{\sigma} \frac{n^{2}}{\varepsilon} / (1 - 3 \log_{\sigma} (1+p))$.
\end{theorem}

\begin{proof}
We want to find $k = 2^{h}$, for an integer $h$, such that the expected number of colliding pairs in automaton $A_{h}$ is at most $\varepsilon$. Then, by Markov's inequality, the probability of having a colliding pair is at most $\varepsilon$. If there are no colliding pairs, then the automaton is prefix-sorted. By Lemma~\ref{lemma:nodes}, if this happens after $h$ doubling and pruning phases, the expected number of nodes in the largest automaton created is at most $N(k) = n(1+p)^{k} + 2$.

By using the bound for the expected number of colliding pairs from Lemma~\ref{lemma:expectation}, we get
\begin{displaymath}
C(k) = \frac{n^{2} (1+p)^{3k}}{\sigma^{k}} \le \varepsilon
\iff
\frac{\log_{\sigma} \frac{n^{2}}{\varepsilon}}{1 - 3 \log_{\sigma} (1+p)} \le k.
\end{displaymath}
As $k$ has to be a power of two, $2 \log_{\sigma} \frac{n^{2}}{\varepsilon} / (1 - 3 \log_{\sigma} (1+p))$ is an upper bound for the smallest suitable $k$. \qed
\end{proof}

\begin{corollary}
For a random reference sequence of length $n$ and mutation rate $p < 0.08$, the expected number of edges in the largest automaton is at most $n(1+p)^{O(\log n)} + O(1)$.
\end{corollary}

\begin{proof}
For $p < 0.08$, the $k$ in Theorem~\ref{theorem:nodes} is at most $3 \log \frac{n^{2}}{\varepsilon}$. By selecting $\varepsilon = \left( \frac{1}{n} \right)^{i}$, we get node bound $n(1+p)^{3 \cdot (2+i) \log n} + 2$ with probability $1 - \varepsilon$. Hence the expected number of nodes is at most
\begin{displaymath}
\begin{split}
n(1+p)^{9 \log n} \sum_{i=0}^{\infty} \left( \frac{(1+p)^{\log n}}{n} \right)^{i} + 2 \le
n(1+p)^{O(\log n)} + 2.
\end{split}
\end{displaymath}
By Lemma~\ref{lemma:edges}, the expected number of edges is at most $(1+p)$ times that. \qed
\end{proof}

\end{document}